\documentclass[journal,twoside]{IEEEtran}
\usepackage{amsmath,amsthm,amsfonts,amssymb,latexsym,extarrows,booktabs,array,arydshln}
\usepackage{fancyhdr,graphicx,tabularx}
\usepackage{multirow}
\usepackage{makecell}
\usepackage{graphicx}
\usepackage{amsfonts}
\usepackage{graphicx}
\usepackage{subfigure}
\usepackage {amsmath}
\usepackage {amssymb}
\usepackage {latexsym}
\usepackage {cite}
\usepackage {subfigure}
\usepackage {url}
\usepackage {stfloats}
\usepackage {mathrsfs}
\usepackage{algorithm}
\usepackage{algorithmic}
\usepackage{xcolor}
\usepackage{longtable}

\usepackage{setspace}
\IEEEoverridecommandlockouts
\newcommand{\ls}[1]
    {\dimen0=\fontdimen6\the\font
     \lineskip=#1\dimen0
     \advance\lineskip.5\fontdimen5\the\font
     \advance\lineskip-\dimen0
     \lineskiplimit=.9\lineskip
     \baselineskip=\lineskip
     \advance\baselineskip\dimen0
     \normallineskip\lineskip
     \normallineskiplimit\lineskiplimit
     \normalbaselineskip\baselineskip
     \ignorespaces
}

\newtheorem{definition}{Definition}
\newtheorem{proposition}{Proposition}

 \newtheorem{assumption}{Assumption}
 \newtheorem{problem}{Problem}
  \newtheorem{theorem}{Theorem}

\begin{document}
 \date{}
 \title{  Towards An Optimal Reliability-Stability Tradeoff  For Secure Coded Uplink Training in  Large-\\Scale  MISO-OFDM  Systems}
  \title{ Optimal Independence-Checking Coding For Secure  Uplink Training in  Large-Scale \\ MISO-OFDM  Systems}
\author{
\IEEEauthorblockN{Dongyang Xu\IEEEauthorrefmark{1}\IEEEauthorrefmark{2}, Pinyi Ren\IEEEauthorrefmark{1}\IEEEauthorrefmark{2}, and James A. Ritcey\IEEEauthorrefmark{3}}\\
 \IEEEauthorblockA{\IEEEauthorrefmark{1}School of Electronic and Information Engineering, Xi'an Jiaotong University, China} \\
\IEEEauthorblockA{\IEEEauthorrefmark{2}Shaanxi Smart Networks and Ubiquitous Access Research Center, China. }\\
\IEEEauthorblockA{\IEEEauthorrefmark{3}Department of Electrical Engineering, University of Washington, USA.}\\
\IEEEauthorblockA{E-mail: \{\emph {xudongyang@stu.xjtu.edu.cn, pyren@mail.xjtu.edu.cn, ritcey@ee.washington.edu}
\}}
%\thanks{
%\IEEEcompsocthanksitem D. Xu  and P. Ren (corresponding author)  are with the Department of Information and Communications Engineering, Xi'an Jiaotong University, Xi'an, Shaanxi 710049 China (E-mail: xudongyang@stu.xjtu.edu.cn,  pyren@mail.xjtu.edu.cn)  D. Xu and P. Ren  are also with National Simulation Education Center for Communications and Information Systems and Shaanxi Smart Networks and Ubiquitous Access Research Center.
%\IEEEcompsocthanksitem J. A. Ritcey are with the Department of Electrical Engineering, University of Washington, Seattle, WA 98195 USA (E-mail: jar7@uw.edu).
%}
}
\maketitle
\begin{abstract}
Due to  the \emph{publicly-known} \emph{deterministic} characteristic of pilot tones, pilot-aware attack, by  jamming,  nulling and spoofing pilot tones,   can significantly paralyze the  uplink channel training  in  large-scale MISO-OFDM systems. To solve this, we in this paper  develop an independence-checking coding based (ICCB)  uplink training  architecture for one-ring scattering scenarios allowing for uniform  linear arrays (ULA) deployment. Here,  we not only insert randomized pilots on subcarriers for  channel impulse response (CIR) estimation, but also  diversify and encode subcarrier activation patterns (SAPs) to convey those pilots simultaneously. The coded SAPs, though interfered by arbitrary unknown SAPs in wireless environment, are qualified  to be reliably  identified and decoded into the original pilots by checking  the  hidden channel independence existing in subcarriers.  Specifically, an independence-checking coding (ICC) theory is formulated to support the encoding/decoding  process in this  architecture.  The optimal ICC code  is further developed  for guaranteeing a well-imposed estimation of CIR  while maximizing the code rate. Based on this code, the identification error probability (IEP) is  characterized to  evaluate   the  reliability of this architecture.  Interestingly, we discover   the principle of   IEP reduction by exploiting the array  spatial correlation, and prove  that zero-IEP, i.e., perfect reliability,  can be  guaranteed  under  continuously-distributed mean angle of arrival (AoA).  Besides this, a novel closed form of  IEP expression is derived  in discretely-distributed case.  Simulation results  finally verify the effectiveness  of the proposed architecture.
\end{abstract}
%\newpage
\section{Introduction}
%\ls{0.77}
\label{introduction}
Security paradigms  in wireless communications has  attracted increasing  attention with the evolution of air interface technology towards the requirements of  future 5G networks. In those envisioned scenarios,   multiple existing technologies, such as orthogonal frequency-division multiplexing (OFDM),  are closely integrated  with novel innovative attempts, such as large-scale multiple-antenna technique  or namely massive multiple-input, multiple-output (Massive MIMO)~\cite{Bogale}. And the phenomenon  accompanied by is  that the  imperishable  characteristic of wireless channels, such as the open and shared nature,  have always been rendering  those air interface technologies  vulnerable to growing  security  attacks, including the denial of service (DoS) attacks and tampering attacks, among others.  As a major manner of DoS attack, jamming attacks, in a variety of behaviors out of control,  have exhibited its  astonishing destructive power on those existing~\cite{Shahriar} and  emerging air interface  techniques~\cite{Pirzadeh}.

A very typical example is that  OFDM systems under large-scale antenna arrays  are very  suspectable to  the protocol-aware attack, a well-directed attack that can sense the specific protocols and intensify the effectiveness of attack  by jamming a physical layer mechanism instead of data payload directly. As a typical protocol-aware attack,  pilot-aware attack  could  hinder the regular channel training between legitimate transceiver pair.   This is done, in theory,  by  jamming/nulling/spoofing  the deterministic pilot tones which are known and shared on the time-frequency resource grid (TFRG) by  all parties for channel acquisition~\cite{Ozdemir,Clancy2,Xu}. That is to say,  pilot-aware attack could embrace three flexible modes, i.e., pilot tone jamming  (PTJ) attack~\cite{Clancy2}, pilot tone nulling  (PTN) attack~\cite{Clancy2} and pilot tone spoofing  (PTS) attack~\cite{Xu}. As  an  example of PTS attack in narrow-band single-carrier systems, pilot contamination (PC)  attack was  first introduced  and analysed  by~\cite{Zhou}. Following~\cite{Zhou}, many research   have been investigated on the advantage of  large-scale multi-antenna arrays on defending against  PC attack~\cite{Kapetanovic1,Tugnait1,Kapetanovic2}. However, those studies were limited to the attack  detection  by exploiting the physical layer information, such as auxiliary training or data sequences~\cite{Kapetanovic1,Tugnait1} and some prior-known channel information~\cite{Kapetanovic2}.

The  first attempt to resolve pilot aware attack was proposed for a conventional OFDM systems in~\cite{Shahriar2}, that is, transforming the PTN and PTS attack into   PTJ attack by  randomizing the  locations and values of regular pilot tones on TFRG. Assuming the independent subcarriers, authors in~\cite{Xu} proposed a frequency-domain subcarrier (FS) channel estimation framework under the PTS attack by exploiting pilot randomization and  independence component analysis (ICA). One key problem is that the practical subcarriers are not mutually independent in the scenarios with limited channel taps, and thus ICA does not apply in this case. What's most important is that the influence of SAPs on CIR estimation was not evaluated. Actually, when the so-called optimal code is adopted, its CIR estimation is extremely ill-imposed and unprecise.

This further motives us to  provide a secure  large-scale multi-antenna OFDM systems, with some necessary consideration, i.e., array spatial correlation, and redesign the overall  pilot sharing  process during the uplink channel training phase.   Most importantly, we have to redesign the supporting CIR estimation process.  Before that,  we have to admit that the pilot randomization technique, though necessary for resolving pilot-aware attack,  brings  to the process two  bottlenecks, i.e., unpredictable  attack modes  and non-recoverable pilot information covered by  random wireless  channels. Basically, an efficient hybrid attack is more likely to be the following:
\begin{problem}[\textbf{Attack Model}]
An attacker in  hybrid attack mode can  choose either  PTJ  mode or silence cheating (SC) mode for  intentional information  hiding as well.
\end{problem}
Two notes should be noticed, that is, 1) an attacker in  PTJ mode  could choose two behaviors, i. e., wide-band  pilot jamming (WB-PJ) attack and partial-band pilot  jamming (PB-PJ) attack. 2) the  attacker in SC mode turns to keep silent  for cheating  the legitimate node.  In this case,  though  the legitimate node adopts random pilots by supposing  the attacker  exists, the attacker actually does not  pay the price since it does not jam at all.

On the other hand, we could  further  identify the  second bottleneck  as follows:
\begin{problem}
Randomized pilots, if utilized for uplink channel training   through wireless channels, cannot be separated, let alone identified.
\end{problem}
 This issue refers to three fundamental concepts which are respectively  recognized  as  pilot  conveying, separation and   identification in this paper. Here,  the innovative methodology  we introduce   is:
 \emph{ Selectively activate and deactivate the OFDM subcar-
riers and create various SAP candidates. Diversify  SAPs to encode pilots and reuse those coded subcarriers carrying pilot information, to estimate the uplink channels simultaneously.}
 In what follows, the main contributions of this paper are  summarized:
\begin{enumerate}
  \item  First,  a deterministic and precise  encoding  principle is established such that  arbitrary SAPs  can be encoded as a binary code. The  ICC theory is then developed to further optimize the code such that  arbitrary two codewords in the code, if being superimposed on each other, can be  separated and identified reliably. Furthermore, an optimal  ICC codebook is formulated  with the maximum code rate while guaranteeing  a well-imposed CIR estimation.  Based on this code,  a reliable ICCB uplink training architecture  is finally built up  by  constructing an one-to-one mapping/demapping relationship between pilots, codewords and  SAPs.
  \item We  further  characterize the reliability of this architecture  as the identification error probability (IEP) and  discover  a hidden  phenomenon that  when  subcarrier estimations are performed on the basis of this architecture,  the array spatial correlation existing in the  subcarriers overlapped from the legitimate node and the attacker can further reduce IEP. At this point, the attacker can actually help the legitimate node to improve  the reliability.  Interestingly, it can also be proved that  zero IEP cannot be achieved only when  the attacker is located in  the  clusters with the same  mean AoA as the legitimate node. This principle, in theory, could facilitate   the acquisition of  the position of Ava.  If we consider the mean AoA with  continuous distribution, the reliability, in this sense, can be perfectly guaranteed. Otherwise,  for a  practical discrete distribution model, we again show how much  the reliability  could be further reinforced.
\end{enumerate}
The rest of the paper is summarized as follows. In section~\ref{PSA}, we present an overview of  pilot-aware attack on  multi-antenna OFDM  systems. In Section~\ref{PUTA}, we introduce an ICCB uplink training architecture. Channel estimation and identification enhancement  is described  in  Section~\ref{CEIE}. Numerical results are presented in Section~\ref{NR} and finally  we conclude our work in Section~\ref{Conclusions}.
\section{Overview of  Pilot-Aware Attack on  Multi-Antenna OFDM  Systems}
\label{PSA}
In this section, we will provide a basic overview  of  pilot-aware attack by introducing  three basic configurations, including  the  system  and  signal model as well as the channel estimation model. Under this background,  we will then review  the influence  of a common-sense  technique, i.e.,  pilot randomization,  on the  pilot-aware attack and identify  the existing key impediments.
\begin{figure}[!t]
\centering \vspace{-10pt}\includegraphics[width=0.65\linewidth]{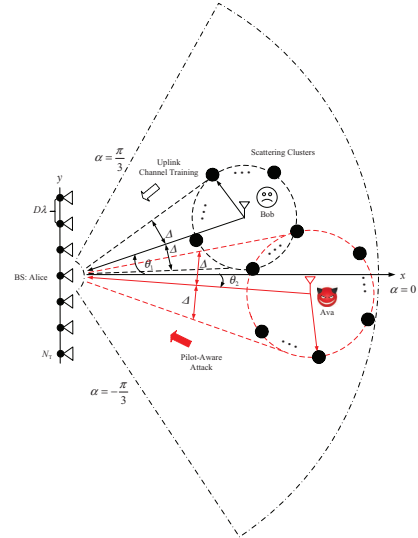}
\caption{Diagram of large-scale MISO-OFDM system under the  wide-band one-ring scattering  model. In this system,  AoA ranges of  Bob and Ava overlap, which incurs an effective pilot-aware attack on the  uplink channel estimation.}\label{fig:system_b}
\vspace{-10pt}
\label{System_model}
\end{figure}
\subsection{System Description}
We consider an  synchronous large-scale MISO-OFDM system with a $N_{\rm T}\gg 1$-antenna base station (named as Alice) and a single-antenna legitimate user (named as Bob). As shown in  Fig.~\ref{System_model},  the based station (BS) with angle spread  $\Delta $ is equipped with a $D\lambda$-spacing directive ULA and  placed at the origin along the $y$-axis to serve a 120-degree sector that is centered around the $x$-axis ($\alpha  = 0$). We assume  that no energy is received for angles $\alpha  \notin \left[ {-\frac{\pi }{3},\frac{{\pi }}{3}} \right]$.   Furthermore, we consider the wide-band one-ring scattering  model for which   Bob  is  surrounded by local scatterers  within  $\left[ {{\theta _1} - {\Delta },{\theta _1} + {\Delta }} \right]$~\cite{Adhikary,Han}. Here $\theta _1$ represents the mean AoA of clusters surrounding Bob.

 In this system,  pilot-tone based uplink channel training process is considered in which  $N$ available subcarriers indexed by $\Psi$ are provided during each available  OFDM symbol time. In principle,   $N_{\rm B}$ subcarriers  indexed by ${\Psi _{\rm{B}}}{\rm{ = }}\left\{ {{i_0},{i_1}, \ldots ,{i_{{N_{\rm B}} - 1}}} \right\}$ are employed  for  pilot tone insertion and  the following channel estimation.  A single-antenna malicious node (named as Ava)  then aims to disturb  this training process  by jamming/spoofing/nulling  those pilot tones. We denote  the set of victim  subcarriers  by ${\Psi _{\rm{A}}}{\rm{ = }}\left\{ {{i_0},{i_1}, \ldots ,{i_{{N_{\rm A}} - 1}}} \right\}$ where $N_{\rm A}$ denotes the number of victim  subcarriers.  Furthermore, we   make the following assumption:
\begin{assumption}
Ava is  surrounded by local scatterers  within $\left[ {{\theta _2} - {\Delta },{\theta _2} + {\Delta }} \right]$ and always has the overlapping AoA intervals with Bob, this is, $\left[ {{\theta _2} - \Delta ,{\theta _2} + \Delta } \right] \cap \left[ {{\theta _1} - \Delta ,{\theta _1} + \Delta } \right] \ne \emptyset $. Here, $\theta _2$ denotes the mean AoA of clusters surrounding Ava.
\end{assumption}
This assumption is supported by the scenario where a common large scattering body (e.g., a large building) could create a set of angles common to all nodes in the system and the overlapping is inevitable. The result is that   the  channel covariance  eigenspaces of two nodes are coupled and  the attack is hard to be eliminated through angular separation~\cite{Adhikary}.
\subsection{Receiving Signal Model}
To begin with, we denote  pilot tones of  Bob and Ava at the $j$-th subcarrier and $k$-th symbol time, respectively by ${x_{\rm{B}}^j\left[ k \right]}, {j \in {\Psi _{\rm{B}}}}$ and  ${x_{\rm{A}}^j\left[ k \right]}, {j \in {\Psi _{\rm{A}}}}$.
\begin{assumption}
 We in this paper assume $x_{{\rm{B}}}^i\left[ k \right] = {x_{{\rm{B}}}}\left[ k \right]= \sqrt {{\rho _{{\rm{B}}}}} {e^{j{\phi_{k} }}},{{i \in {\Psi _{\rm{B}}}}}$ for low overhead consideration and theoretical analysis. Alternatively,  we can superimpose ${x_{{\rm{B}}}}\left[ k \right] $ onto a dedicated  pilot sequence optimized under a non-security oriented scenario. At this point,  $\phi_{k}$ can be an additional phase difference for security consideration. We do not constraint the  strategies of pilot tones of Ava such that ${x^i_{\rm{A}}}\left[ k \right] = \sqrt {{\rho_{\rm{A}}}} {e^{j{{\varphi}_{k,i}} }},{{i \in {\Psi _{\rm{A}}}}}$.
\end{assumption}

Let us proceed to the basic OFDM procedure. First,  the frequency-domain pilot signals of Bob and Ava  over $N$ subcarriers  are  respectively  stacked as $N$ by $1$  vectors ${{\bf{x}}_{\rm{B}}}\left[ k \right] = \left[ {{x_{{\rm{B}},j}}\left[ k \right]} \right]_{{ {j \in {\Psi}}}}^{\rm{T}}$ and ${{\bf{x}}_{\rm{A}}}\left[ k \right] = \left[ {{x_{{\rm{A}},j}}\left[ k \right]} \right]_{{ {j \in {\Psi}}}}^{\rm{T}}$. Here  there exist:
  \begin{equation}\label{E.1}
{x_{{\rm{B}},j}}\left[ k \right] = \left\{ {\begin{array}{*{20}{c}}
{x_{\rm{B}}\left[ k \right]}&{{j \in {\Psi _{\rm{B}}}}}\\
0&{{j \notin{\Psi _{\rm{B}}}}}
\end{array}} \right., {x_{{\rm{A}},j}}\left[ k \right] = \left\{ {\begin{array}{*{20}{c}}
{x_{\rm{A}}^j\left[ k \right]}&{{j \in {\Psi _{\rm{A}}}}}\\
0&{{j \notin{\Psi _{\rm{A}}}}}
\end{array}} \right.
\end{equation}
Assume that the  length of cyclic prefix is larger than  $L$.  The parallel streams, i.e.,  ${{\bf{x}}_{{\rm{B}}}}\left[ k \right]$ and  ${{\bf{x}}_{{\rm{A}}}}\left[ k \right]$, are modulated with inverse fast Fourier transform (IFFT). After removing the cyclic prefix at the $i$-th receive antenna and $k$-th OFDM symbol time, Alice derive the time-domain $N$ by $1$  vector ${{\bf{y}}^i}\left[ k \right]$  as:
\begin{equation}\label{E.3}
{{\bf{y}}^i}\left[ k \right] = {\bf{H}}_{{\rm{C,B}}}^i{{\bf{F}}^{\rm{H}}}{{\bf{x}}_{\rm{B}}}\left[ k \right] + {\bf{H}}_{{\rm{C,A}}}^i{{\bf{F}}^{\rm{H}}}{{\bf{x}}_{\rm{A}}}\left[ k \right] + {{\bf{v}}^i}\left[ k \right]
\end{equation}
where ${\bf{H}}_{{\rm{C,B}}}^i$ and ${\bf{H}}_{{\rm{C,A}}}^i$ are $N \times N$ circulant matrices for which the  first column of  ${\bf{H}}_{{\rm{C,B}}}^i$  and ${\bf{H}}_{{\rm{C,A}}}^i$ are respectively given by ${\left[ {\begin{array}{*{20}{c}}
{{\bf{h}}_{\rm{B}}^{{i^{\rm{T}}}}}&{{{\bf{0}}_{1 \times \left( {N - L} \right)}}}
\end{array}} \right]^{\rm{T}}}$ and ${\left[ {\begin{array}{*{20}{c}}
{{\bf{h}}_{\rm{A}}^{{i^{\rm{T}}}}}&{{{\bf{0}}_{1 \times \left( {N - L} \right)}}}
\end{array}} \right]^{\rm{T}}}$. Here, ${\bf{h}}_{\rm{B}}^i\in {{\mathbb C}^{L \times 1}} $ and ${\bf{h}}_{\rm{A}}^i \in {{\mathbb C}^{L\times 1}}$  are CIR  vectors, respectively from Bob and Ava to the $i$-th receive antenna of Alice. ${\bf{h}}_{\rm{A}}^i $ is   assumed to be independent with  ${\bf{h}}_{\rm{B}}^i$. ${{\bf{v}}^i}\left[ k \right]\in {{\mathbb C}^{N \times 1}}$ with ${{\bf{v}}^i}\left[ k \right] \sim {\cal C}{\cal N}\left( {0,{{{\bf{I}}_N}\sigma ^2}} \right)$  is the  AWGN  vector  at the $i$-th antenna and $k$-th symbol time
%For the $l$-th path of $i$-th antenna,  Bob has the  channel power delay profile (PDP) denoted by $\sigma _{{\rm{B}},l,i}^2$ and  Ava owns the PDP denoted by $\sigma _{{\rm{A}},l,i}^2$. ${\bf{F}}$ denotes the $N\times N$ unitary DFT matrix and ${{\bf{v}}^i_{N}}\left[ k \right]\sim{\cal C}{\cal N}\left( {0,{\sigma ^2}{{\bf{I}}_{N \times N}}} \right)$ denotes the complex  additive white Gaussian noise (AWGN) vector where   $\sigma ^2$ represents the average noise power.
Taking  FFT,  Alice  finally derives the    frequency-domain $N$ by $1$ signal vector  at the $i$-th receive antenna and $k$-th OFDM symbol time as
\begin{equation}\label{E.4}
{\widetilde {\bf{y}}^i}\left[ k \right] = {\rm{diag}}\left\{ {{{\bf{x}}_{\rm{B}}}\left[ k \right]} \right\}{{\bf{F}}_{\rm{L}}}{\bf{h}}_{\rm{B}}^{{i}}  + {\rm{diag}}\left\{ {{{\bf{x}}_{\rm{A}}}\left[ k \right]} \right\}{{\bf{F}}_{\rm{L}}}{\bf{h}}_{\rm{A}}^{{i}} + {{\bf{w}}^i_{N}}\left[ k \right]
\end{equation}
Here, there exists ${{\bf{F}}_{\rm{L}}} = \sqrt N {\bf{F}}\left( {:,1:L} \right)$ where $\bf {F}$ denotes the DFT matrix.  And we have  ${{\bf{w}}^i_{j}}\left[ k \right]={{\bf{F}}_{j}}{{\bf{v}}^i}\left[ k \right]$ where ${{\bf{F}}_{j}}$ is the $j$-row submatrix of $\bf {F}$.
Throughout this paper, we assume that the CIRs belonging to  different paths  at each antenna  exhibit spatially uncorrelated Rayleigh fading. We denote power delay profiles (PDPs)  of  the $l$-th path of  Bob and Ava, respectively by $\sigma _{{\rm{B}},l}^2$, $\sigma _{{\rm{A}},l}^2$.
Without loss of generality, each path has the uniform and normalized PDP satisfying  $\sum\limits_{l = 1}^L {\sigma _{{\rm B},l}^2}  = 1, \sum\limits_{l = 1}^L {\sigma _{{\rm A},l}^2}  = 1$~\cite{McKay}.  For each path, CIRs  of different antennas  are assumed to be spatially correlated.
 In  one-ring scattering scenarios, the correlation  between the channel coefficients of antennas $1 \le m,n \le N_{\rm T}$, $\forall l$ can be  defined by~\cite{Adhikary}:
    \begin{equation}\label{E.5}
{\left[ {{{\bf{R}}_{{k}}}} \right]_{m,n}} = \frac{1}{{2\Delta }}\int_{ - {\Delta} + {\theta_{k}}}^{{\Delta} + {\theta_{k}}} {{e^{ - j2\pi D\left( {m - n} \right)\sin \left( \theta  \right)}}} d\theta, k=1,2
  \end{equation}
Here, ${{\bf{R}}_{{i}}}$ represents the channel covariance matrix of Bob if $i=1$  and  Ava otherwise.   ${{\bf{R}}_{{1}}}$, instead of ${{\bf{R}}_{{2}}}$,   is known by Alice.
\subsection{ Channel Estimation Model}
Now let us turn to describe the estimation models  of  FS channels under specific  attacks.  First,   Ava under PTS attack mode could learn  the pilot tones employed by Bob  in advance and  impersonate Bob  by utilizing the same pilot tone learned.  In this case, there exists $
{\Psi _{\rm{B}}}\cup {\Psi _{\rm{A}}} = {\Psi _{\rm{B}}}$ and $x_{{\rm{A}}}^i\left[ k \right] = {x_{{\rm{B}}}}\left[ k \right],{i \in {\Psi _{\rm{B}}}}$. Signals in  Eq.~(\ref{E.4}) can be rewritten as:
\begin{equation}\label{E.6}
{\widetilde {\bf{y}}^i_{\rm PTS}}\left[ k \right] = {{\bf{F}}_{\rm{L}}}{\bf{h}}_{\rm{B}}^{{i}}{x_{\rm{B}}}\left[ k \right]  + {{\bf{F}}_{\rm{L}}}{\bf{h}}_{\rm{A}}^{{i}}{x_{\rm{B}}}\left[ k \right] + {{\bf{w}}^i_{N}}\left[ k \right]
\end{equation}

Finally, a least square (LS)  based channel estimation is formulated by:
${\widehat {{\bf{h}}}_{con}^i} = {\bf{h}}_{\rm{B}}^i + {\bf{h}}_{\rm{A}}^i + {\left( {{{\bf{F}}_{\rm{L}}}} \right)^ + }\frac{{x_{\rm{B}}^{\rm{H}}\left[ k \right]}}{{{{\left| {x_{\rm{B}}^{\rm{H}}\left[ k \right]} \right|}^2}}}{{\bf{w}}^i_{N}}\left[ k \right]$
where $\left( {{{\bf{F}}_{\rm{L}}}} \right)^+ $ is  the MoorePenrose
pseudo-inverse of $ {{{\bf{F}}_{\rm{L}}}} $.  We  see that   the estimation of ${\bf{h}}_{\rm{B}}^i $  is contaminated by  ${\bf{h}}_{\rm{A}}^i $   with  a noise bias when a PTS  attack happens.

As to PTN  attack,  we emphasize the difference lying in the fact that  there exists  ${\rm{diag}}\left\{ {{{\bf{x}}_{\rm{A}}}\left[ k \right]} \right\}{\rm{ = }}{\bf{\Sigma }} \odot {x_{\rm{A}}}\left[ k \right]$ such that ${\bf{\Sigma }}{{\bf{F}}_{\rm{L}}}{\bf{h}}_{\rm{A}}^i{x_{\rm{A}}}\left[ k \right] =  - {{\bf{F}}_{\rm{L}}}{\bf{h}}_{\rm{B}}^i{x_{\rm{B}}}\left[ k \right]$. Obviously, Ava can derive a unique solution of the diagonal matrix  ${\bf{\Sigma}}$ because the assumed  Ava can get both ${\bf{h}}_{\rm{B}}^i $ and ${\bf{h}}_{\rm{A}}^i $ (a very strong assumption in~\cite{Clancy2}).  In this case, the received signals can be rewritten as ${\widetilde {\bf{y}}^i_{\rm PTN}}\left[ k \right] =  {{\bf{w}}^i_{N}}\left[ k \right]$.
We see that the received signals are completely random noises, which can be seen as the worst destruction.

In order to represent the case where PTJ attack happens,   we configure  the matrix ${\bf{\Sigma}}$ with  random input values.   The estimated channels are with the similar form as those in PTS attack. The difference is that unlike PTS attack, the estimated channels  cannot  benefit both Bob and Ava, which is least efficient.
\subsection{ Influence of Pilot Randomization on   Pilot-Aware Attack}
To defend  against  pilot-aware attack, the commonsense is that  Bob shall  randomize its own  pilot tones.  In practice,  the randomization of pilot tone values is employed.  More specifically, each of the candidate pilot phases   is  mapped into a unique quantized  samples, chosen  from the set ${\cal A}$, defined by ${\cal C} = \left\{ {{\phi}:{{{\phi} = 2m\pi } \mathord{\left/
 {\vphantom {{{\phi _k} = 2m\pi } C}} \right.
 \kern-\nulldelimiterspace} C},0 \le m \le C - 1} \right\}$ where $C$ reflects  the quantization resolution.  This  type of pilot randomization, due to the constraint  of discrete phase samples,  practically could not prevent  a hybrid  attack from happening but serves as a prerequisite for defending against pilot aware attack.

 In what follows, we make the following assumption for Bob for the sake of  theoretical analysis:
 \begin{assumption}
 During  two adjacent OFDM symbol time, such as, $k_i,k_{i+1}, i\ge0$,   two pilot phases ${\phi _{{k_i}}}$ and ${\phi _{{k_{i+1}}}}$  are kept with fixed phase difference, that is,   ${\phi _{{k_{i+1}}}} - {\phi _{{k_{i}}}} = \overline \phi $. Here, ${\phi _{{k_{i+1}}}}$ and ${\phi _{{k_{i}}}}$ are both random but $\overline \phi$  are deterministic and  publicly known.
 \end{assumption}
Institutively,  how the value $C$ increases affects  the  performance of   anti-attack technique. However, things seem not to be simple as we think.  As discussed in the Introduction part, a fact is that randomized pilots, if utilized for uplink channel training   through wireless channels, cannot be separated, let alone identified.
\begin{figure}[!t]
\centering \vspace{-10pt}\includegraphics[width=1.0\linewidth]{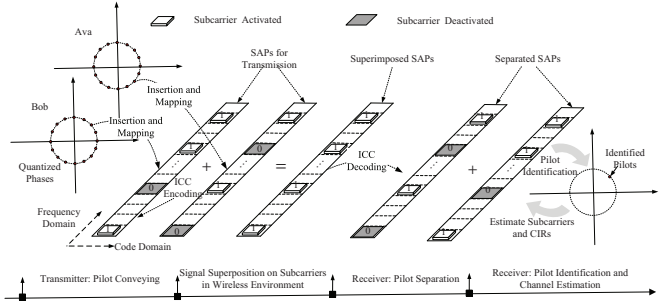}
\caption{Diagram of ICCB uplink channel  training procedures.}\label{fig:system_b}
\label{ICCB}
\end{figure}

 \section{ICCB Uplink Training Architecture}
 \label{PUTA}
In view of above issues, we in  this section  aim to construct a novel  pilot sharing mechanism, logically including three key procedures, i.e., pilot conveying, pilot separation and pilot identification. Each procedure can be found in Algorithm 1 and Fig.~\ref{ICCB}.
\subsection{Pilot Conveying via Binary Code on Code-Frequency Domain}
\label{PCBC}
\subsubsection{Binary Code}
The  Eq. (40) in~\cite{Shakir}  provides  a  decision threshold  function $\gamma  \buildrel \Delta \over =f\left( {{N_{\rm{T}}},{P_f}} \right)$, for measuring how many  antennas on one subcarrier  are required to achieve a certain probability ${P_f}$ of false alarm. Here  we consider three symbol time and  a $3 \times N_{\rm T}$ receiving signal matrix is created for detection.   Under this requirement, we try to build up a relationship between SAPs with the common binary code.  Before that, we have the following definition:
 \begin{definition}
 One subcarrier  can be precisely encoded if, for any $\varepsilon > 0$, there exists a positive number $\gamma \left( \varepsilon  \right)$  such that, for all  $\gamma  \ge \gamma \left( \varepsilon  \right)$, ${P_f}$  is smaller than $\varepsilon $.
\end{definition}
We should note that $f\left( {{N_{\rm{T}}},{P_f}} \right)$ is a monotone decreasing  function of two independent variables ${N_{\rm{T}}}$ and ${P_f}$. For a given probability constraint ${\varepsilon}^*$, we could always expect a lower bound $\gamma \left( {{\varepsilon ^*}} \right)$ of possible thresholds  such that $\gamma \left( {{\varepsilon ^*}} \right) = f\left( {{N_{\rm{T}}},{\varepsilon ^*}} \right)$ is satisfied. Under this equation, we could flexibly configure  ${N_{\rm{T}}}$ and $\gamma \left( {{\varepsilon ^*}} \right)$ to make ${\varepsilon}^*$  approach zero~\cite{Shakir}.   We also find that  the $\gamma$  that achieves zero-${P_f}$  is decreased with the increase of  antennas. Basically, this phenomenon originates from the fact that  the increased dimension  makes   the eigenvalues of noise matrix  to be more concentrated in a narrow interval, which is determined by the well-known Marcenko-Pastur Law~\cite{Hoydis}.
\subsubsection{Code Frequency Domain}
Based on  Definition 1, we can encode  the $ m$-th subcarrier as  a binary  digit   ${s_{ m}} $   according to:
 \begin{equation}\label{E.18}
{s_{m}} = \left\{ {\begin{array}{*{20}{c}}
1&\rm {if \,\,\,there \,\,\, exist  \,\,\, signals}\\
0&{otherwise}
\end{array}} \right.
  \end{equation}
Meanwhile, let us denote a set of binary code vectors  by ${{\cal S}}$  with ${{\cal S}}= \left\{ {\left. {\bf{s}} \right|{s_{m}} \in \left\{ {0,1} \right\},1 \le m \le {{L_s}}} \right\}$ where ${L_s}$ denotes the maximum length of  the  code.  Then, a code frequency domain could be constructed as a set of pairs $\left( {{{\bf{s}}},b} \right)$ with ${\bf{s}} \subset {{\cal S}}$ and $1\le b \le N_{\rm B}$ where $b$ is an integer  representing the subcarrier  index of appearance of the code. This can be depicted in Fig.~\ref{ICCB}.
\subsubsection{Binary Codebook Matrix}
On the formulated  code-frequency domain, we group  the binary digits and  construct  the binary code  by presenting  a binary codebook as follows:
\begin{definition}\label{E.19}
Given a $N_{\rm B} \times C$ binary matrix ${\bf C}$ with each element satisfying  ${c_{i,j}} \in {\bf{s}}\subset {{\cal S}}$, we denote the $i$-th column of ${\bf C}$  by ${\bf c}_{i}$ with ${{\bf{c}}_i} = {\left[ {\begin{array}{*{20}{c}}
{{c_{1,i}}}& \cdots &{{c_{N_{\rm B},i}}}
\end{array}} \right]^{\rm{T}}}$.   We call  $\bf C$  a binary codebook matrix  and ${{\bf{c}}_i}$ a codeword  of $\bf C$ of  length $N_{\rm B}$.
\end{definition}
Based on this definition, we also define  a superposition principle between codewords by the following:
\begin{definition}
The superposition sum ${\bf{z}} = {\bf{x}}{{V}}{\bf{y}}$ (designated as the digit-by-digit Boolean sum) of two codewords denoted by ${\bf{x}} = \left( {{x_1},{x_2}, \ldots ,{x_{N_{\rm B}}}} \right)\subset  {\bf C}$  and ${\bf{y}} = \left( {{y_1},{y_2}, \ldots ,{y_{N_{\rm B}}}} \right)\subset  {\bf C}$ is defined by:
  \begin{equation}\label{E.20}
{z_i} = \left\{ {\begin{array}{*{20}{c}}
0&{if\,\,{x_i} = {y_i} = 0}\\
1&{otherwise}
\end{array}} \right.
  \end{equation}
  where ${z_i}$ denotes the $i$-th element of vector ${\bf{z}}$.
\end{definition}
Based on above preparations, the pilot conveying process can be shown in Algorithm 1.
\begin{algorithm}[!t]\label{CEAIE}
\caption{Pilot Conveying, Separation and Identification}
\begin{algorithmic}[1]
\STATE  \textbf{Pilot  Conveying} 1) Insert one  phase  that is  selected from set $\cal A$, onto  subcarriers at the initial OFDM symbol, for instance defined by $k_{0}$. The  phases  of pilot signals inserted in adjacent OFDM symbols, such as $k_{i},i\ge1$ obey the Assumption 3.

2) Construct an one-to-one   mapping from  the  phases in set $\cal A$ to codewords of binary codebook matrix derived in Section~\ref{PCBC}, and then further  to SAPs. Select one phase, i.e., the phase at $k_{0}$,  for  pattern activation. The specific principle is that pilot signals are transmitted on the $i$-th subcarrier if  the $i$-th digit of the codeword  is equal to 1, otherwise  this subcarrier is kept  unoccupied.
\STATE  \textbf{Pilot  Separation} Alice  detects  the available subcarriers to acquire the superimposed SAPs using the detection technique shown in~\cite{Shakir}. Then Alice decodes those superimposed SAPs and derives two individual codewords  by using the inner-product based differential decoding proposed in~\cite{Xu}.
\STATE  \textbf{Pilot  Identification} Separated codewords that satisfy Theorem 1 are qualified to be  identified and then demapped into the  pilot phases in $\cal A$  for  recovering the pilot signals of Bob.
\end{algorithmic}
\end{algorithm}
\subsection{Pilot Separation and Identification  Via ICC}
\label{PSI}
 The study of how to optimize the previous binary  codebook  such that it can  separate and identify   codewords from the disturbed codeword is called  ICC theory, including the encoding principle and decoding principle.
\subsubsection{Encoding Principle}
We  introduce the concept of $s$-overlapping code with  constant weight $w$  by defining:
\begin{definition}
A $N_{\rm B} \times C$ binary matrix $\bf C$  is called a ICC-$\left( {N_{\rm B},s} \right)$ code of length $N_{\rm B}$ and order $s$,  if for any column set $\cal Q$ such that  $ \left| \cal Q \right|=2$,  there exist at least a set $\cal S$ of $s$ rows  such that ${c_{i,j}} = 1,\forall i,j, {i \in \cal S}, j \in  {\cal Q}$.
\end{definition}
In this principle, we  can know that any two codewords in $\bf C$ must overlap with each other on  at least  $s$ non-zero digits.   Backing to the subcarriers, $s$ means the overlapped subcarriers which are exploited for channel estimation.  Now, we try to establish the relationship of  $s$ with the weight $w$ of the code since the number of $w$ determines the code rate.
\begin{theorem}
The weight of ICC-$\left( {N_{\rm B},s} \right)$ code of length $N_{\rm B}$ and order $s$ satisfies $w = \frac{{N_{\rm B} + s}}{2}$ with $N_{\rm B}\ge s$. $w$ is an integer smaller than $N_{\rm B}$.
\end{theorem}
\begin{proof}
See proof in Appendix~\ref{Theorem1}.
\end{proof}
Here and in the following section, we assume the ratio of two integers is always  kept to be an integer without loss of generality.  Based on the theorem, we can derive the number of  codewords or namely the columns  in $\bf C$,  by a binomial coefficient $C = \left( {\begin{array}{*{20}{c}}
{{N_{\rm B}}}\\ {\frac{{{N_{\rm B}} + s}}{2}}
\end{array}} \right)$.   Therefore, we have the following proposition:
\begin{proposition}
 The code rate of ICC-$\left( {N_{\rm B},s} \right)$ code,  defined by ${R_{ICC}} = \frac{{{{\log }_2}\left( C \right)}}{{{N_{\rm B}}}}$, is calculated as
\begin{equation}\label{E.21}
{R_{ICC}} = {\log _2}{\left[ {\frac{{{N_{\rm B}}!}}{{\left( {\frac{{{N_{\rm B}} + s}}{2}} \right)!\left( {\frac{{{N_{\rm B}} -s}}{2}} \right)!}}} \right]^{{1 \mathord{\left/
 {\vphantom {1 {{N_{\rm B}}}}} \right.
 \kern-\nulldelimiterspace} {{N_{\rm B}}}}}}
\end{equation}
\end{proposition}
\begin{theorem}
 The optimal ICC-$\left( {N_{\rm B},s} \right)$ code  maximizing the code rate holds when $s=L$. In this case, the  reliability measured by IEP is given by
\begin{equation}\label{E.25}
{{P}_{\rm{I}}} = \frac{{{N_{\rm B}}! - \left( {\frac{{{N_{\rm B}} + L}}{2}} \right)!\left( {\frac{{{N_{\rm B}} - L}}{2}} \right)!}}{{{2^{{N_{\rm B}} + 1}}\left( {\frac{{{N_{\rm B}} + L}}{2}} \right)!\left( {\frac{{{N_{\rm B}} - L}}{2}} \right)!}}
\end{equation}
\end{theorem}
\begin{proof}
See proof in Appendix~\ref{Theorem2}.
\end{proof}

\subsubsection{Decoding Procedure}
The related  technique in this part is same with that in Fig. 3 of~\cite{Xu}.  We do not specify this.

The overall process can be shown in Algorithm 1.

\section{Channel Estimation and Identification Enhancement}
\label{CEIE}
\label{Estimation}
 In this section,  we continue  our design work  and focus on the channel estimation phase.
 %New questions emerged will be answered further
%\begin{question}
%How to estimate FS channels based on the identified pilots?
% Is it possible to  improve the reliability  performance of ICC theory by further digging  channel properties ?
%\end{question}
%We begin our discussion by talking about the FS channel estimation to answer Question 1.
\subsection{FS Channel Estimation}
We do not consider the case where there is no attack since in this case LS estimator is a natural choice.  If  looking back to the pilot identification under a certain attack, we could derive two results, that is, one identified Bob's pilot vector or two confusing pilot vectors. For better considering the two cases,  we in this section assume the identification error happens and forget the case without error, that is,  we could get  two confusing pilot vectors  defined by  ${\bf{x}}_{{\rm{L,1}}}{\rm{ = }}{\left[ {\begin{array}{*{20}{c}}
{{x_{\rm{B}}}\left[ {{k_0}} \right]}&{{x_{\rm{B}}}\left[ {{k_1}} \right]}
\end{array}} \right]^{\rm{T}}}$ and ${\bf{x}}_{{\rm{L,2}}}{\rm{ = }}{\left[ {\begin{array}{*{20}{c}}
{{x_{\rm{A}}}\left[ {{k_0}} \right]}&{{x_{\rm{A}}}\left[ {{k_1}} \right]}
\end{array}} \right]^{\rm{T}}}$ within two OFDM symbol time, i.e., ${k_0}$ and ${k_1}$. In this way,  the  estimator to be designed in this case can also apply in the another case naturally.

We consider two OFDM symbol time, i.e., ${k_0}$ and ${k_1}$ and $s,s\ge1 $ randomly-overlapping subcarriers.  The randomness means their random positions of carrier frequency.  Then the signals received on overlapping  subcarriers within ${k_0}$ and ${k_1}$  are stacked  as the ${2 \times  {N_{\rm T}s}}$ matrix ${{\bf{Y}}_{\rm L}}$, equal to
\begin{equation}\label{E.26}
{{\bf{Y}}_{\rm L}} = {{\bf{X}}_{{\rm L}}}{{\bf{H}}_{{{{\rm L}}}}} + {{\bf{N}}_{\rm L}}
 \end{equation}
 where  ${{\bf{X}}_{{\rm L}}}$  is denoted by a ${2 \times 2}$ matrix satisfying  ${{\bf{X}}_{\rm L}} = \left[ {\begin{array}{*{20}{c}}
{\bf{x}}_{{\rm{L,1}}}&{\bf{x}}_{{\rm{L,2}}}
\end{array}} \right]$. The integrated ${2 \times  {N_{\rm T}s}}$  channel matrix ${{\bf{H}}_{{{{\rm L}}}}}$  satisfies    ${{\bf{H}}_{{{{\rm L}}}}} = {\left[ {\begin{array}{*{20}{c}}
{{\bf{h}} _{{\rm{B}},{\rm L}}^{\rm{T}}}&{{\bf{h}} _{{\rm{A}},{\rm L}}^{\rm{T}}}
\end{array}} \right]^{\rm{T}}}$. Here, there exist ${{\bf{h}}_{{\rm{B}},{\rm L}}} = \left[ {\begin{array}{*{20}{c}}
{{{\left( {{{\bf{F}}_{{\rm{L}}, s}}{\bf{h}}_{\rm{B}}^i} \right)}^{\rm{T}}}}&{, \ldots ,}&{{{\left( {{{\bf{F}}_{{\rm{L}}, s}}{\bf{h}}_{\rm{B}}^{{N_{\rm{T}}}}} \right)}^{\rm{T}}}}
\end{array}} \right]$ and ${{\bf{h}}_{{\rm{A}},{\rm L}}} = \left[ {\begin{array}{*{20}{c}}
{{{\left( {{{\bf{F}}_{{\rm{L}}, s}}{\bf{h}}_{\rm{A}}^i} \right)}^{\rm{T}}}}&{, \ldots ,}&{{{\left( {{{\bf{F}}_{{\rm{L}}, s}}{\bf{h}}_{\rm{A}}^{{N_{\rm{T}}}}} \right)}^{\rm{T}}}}
\end{array}} \right]$. ${{\bf{F}}_{{\rm{L}}, s}}$ is  the $s$-row matrix for which each index of $s$ rows belongs to the set ${\cal P}_{s}$ that is  defined by   ${\cal P}_{s}=\left\{ {j_1, \ldots ,j_s} \right\}$, $ {{\cal P}_{s}} \subseteq \Psi, \left| {\cal P}_{s} \right| = s$.
${\bf{N}}_{\rm L}$ represents the ${2 \times  {N_{\rm T}s}}$ noise matrix  with  ${{\bf{N}}_{\rm L}} = {\left[ {\begin{array}{*{20}{c}}
{{\bf{w}}_{\rm L}^{\rm{T}}\left[ {{k_0}} \right]}&{{\bf{w}}_{\rm L}^{\rm{T}}\left[ {{k_1}} \right]}
\end{array}} \right]^{\rm{T}}}$ where there exists  ${{\bf{w}}_{\rm L}}\left[ k \right] = \left[ {\begin{array}{*{20}{c}}
{{{\bf{w}}_{s}^{{1^{\rm{T}}}}}\left[ k \right]}&{, \ldots ,}&{{{\bf{w}}_{s}^{{N_{\rm{T}}}^{\rm{T}}}}\left[ k \right]}
\end{array}} \right]$ for $k=k_{0}, k_{1}$.

Now we turn to the procedure of channel estimation.   First, ${\bf{x}}_{{\rm{L,1}}}$ and ${\bf{x}}_{{\rm{L,2}}}{\rm{ }}$,  are deemed  as the candidate weight vectors for estimating. We then consider the sample covariance matrix  given by ${{\bf{C}}_{{{\bf{Y}}_{\rm L}}}} = \frac{1}{{{N_{\rm{T}}}s}}{{\bf{Y}}_{\rm L}}{\bf{Y}}_{\rm L}^{\rm{H}}$ and  finally derive  the  asymptotically-optimal linear minimum mean square error (LMMSE) estimators as
${{\bf{W}}_{{\rm{B}},{\rm{L}}}} = {T_{\rm{B}}} {\bf{x}}_{{\rm{L,1}}}^{\rm{H}}{\bf{C}}_{{{\bf{Y}}_{\rm{L}}}}^{ - 1}$ and
${{\bf{W}}_{{\rm{A}},{\rm{L}}}} = {T_{\rm{A}}} {\bf{x}}_{{\rm{L,2}}}^{\rm{H}}{\bf{C}}_{{{\bf{Y}}_{\rm{L}}}}^{ - 1},$
  where ${T_{\rm{B}}} \buildrel \Delta \over = \frac{{{\rm{Tr}}\left( {{{\bf{R}}_{{\rm{B}},{\rm{L}}}}} \right){\rm{Tr}}\left( {{{\bf{R}}_{\rm{F}}}} \right)}}{{{N_{\rm{T}}}s}}$ and ${T_{\rm{A}}} \buildrel \Delta \over = \frac{{{\rm{Tr}}\left( {{{\bf{R}}_{{\rm{A}},{\rm{L}}}}} \right){\rm{Tr}}\left( {{{\bf{R}}_{\rm{F}}}} \right)}}{{{N_{\rm{T}}}q}}$. Here,  there exists ${\rm{Tr}}\left( {{{\bf{R}}_{{\rm{B,L}}}}} \right) = {\rm{Tr}}\left( {{{\bf{R}}_{{\rm{A,L}}}}} \right)=N_{\rm T}$ and therefore we could define ${T_{\rm{B}}} = {T_{\rm{A}}} = T$.

  The estimated versions of FS channels are  respectively  derived  as
   \begin{equation}\label{E.27}
   {\widehat {\bf{h}}_{{\rm{B}},{\rm{L}}}}{\rm{ = }}{{\bf{W}}_{{\rm{B}},{\rm{L}}}}{{\bf{Y}}_{\rm{L}}}, {\widehat {\bf{h}}_{{\rm{A}},{\rm{L}}}}{\rm{ = }}{{\bf{W}}_{{\rm{A}},{\rm{L}}}}{{\bf{Y}}_{\rm{L}}}
    \end{equation}
  The normalized mean square error (NMSE)  for  the two estimations are respectively defined by $
\varepsilon _{\rm{B}}^2 = \frac{{{\mathbb E}\left\{ {{{\left\| {{{\widehat {\bf{h}}}_{\rm{B,L}}} - {{\bf{h}}_{\rm{B,L}}}} \right\|}^2}} \right\}}}{{{N_{\rm{T}}}s}},\varepsilon _{\rm{A}}^2 = \frac{{{\mathbb E}\left\{ {{{\left\| {{{\widehat {\bf{h}}}_{\rm{A,L}}} - {{\bf{h}}_{\rm{A,L}}}} \right\|}^2}} \right\}}}{{{N_{\rm{T}}}s}}$
Furthermore, the relationship between the ideal channels  with estimated versions can be given by
 ${{\bf{h}}_{\rm{B,L}}} = {\widehat {\bf{h}}_{\rm{B,L}}} + {\varepsilon _{\rm{B}}}{\bf{h}}$ and ${{\bf{h}}_{\rm{A,L}}} = {\widehat {\bf{h}}_{\rm{A,L}}} + {\varepsilon _{\rm{A}}}{\bf{h}}^{'}$
 where $ {\varepsilon _{\rm{B}}}{\bf{h}}$  is uncorrelated with ${{\bf{h}}_{\rm{B,L}}}$ and $ {\varepsilon _{\rm{A}}}{\bf{h}}^{'}$  is uncorrelated with ${{\bf{h}}_{\rm{A,L}}}$. Here, the entries of ${\bf{h}}$  and ${\bf{h}}^{'}$ are i.i.d zero-mean complex Gaussian vectors with each  element having unity variance.

 Based on above results, we could have the following proposition:
  \begin{proposition}
  With the large-scale antenna  array, there exists $\varepsilon _{\rm{B}}^2 = \varepsilon _{\rm{A}}^2$ at  high SNR .
    \end{proposition}
 \begin{IEEEproof}
See proof in Appendix~\ref{appendices_pro2}
 \end{IEEEproof}

\subsection{Identification Enhancement}
Identification enhancement in this section means reducing IEP further. Since Bob could  get two confusing pilots and two confusing estimated channels, we model the process of identification enhancement  as a decision between two hypothesis:
 \begin{equation}\label{E.33}
 \begin{array}{l}
{{\cal H}_{\rm{0}}}:{\widehat {\bf{h}}_{{\rm{B}},{\rm{L}}}} \to Bob,
{\cal {H}_{\rm{1}}}:{\widehat {\bf{h}}_{{\rm{A}},{\rm{L}}}} \to Bob
\end{array}
\end{equation}
 For the sake of simplicity, we define the following  eigenvalue decomposition:
\begin{equation}\label{E.29}
{{\bf{R}}_i} = {{\bf{U}}_i}{{\bf{\Lambda }}_i}{\bf{U}}_i^{\rm{H}},{{\bf{\Lambda }}_i} = diag\left\{ {\begin{array}{*{20}{c}}
{{\lambda _{i,1}}}& \cdots &{{\lambda _{i,{\rho _i}}}}&0& \cdots &0
\end{array}} \right\}
\end{equation}
\begin{equation}\label{E.30}
{\overline {\bf{R}} _i} = {{\bf{U}}_i}{\overline {\bf{\Lambda }} _i}{\bf{U}}_i^{\rm{H}},{\overline {\bf{\Lambda }} _i} = diag\left\{ {\begin{array}{*{20}{c}}
{\lambda _{i,1}^{ - 1}}& \cdots &{\lambda _{i,{\rho _i}}^{ - 1}}&0& \cdots &0
\end{array}} \right\}
\end{equation}
\begin{equation}\label{E.31}
{{\bf{R}}_{\rm{F}}} = {{\bf{V}}_{\rm{f}}}{{\bf{\Sigma }}_{\rm{f}}}{\bf{V}}_{\rm{f}}^{\rm{H}},{{\bf{\Sigma }}_{\rm{f}}} = diag\left\{ {\begin{array}{*{20}{c}}
{{\lambda _{{\rm{f}},1}}}& \cdots &{{\lambda _{{\rm{f}},{\rho _{{\rm{f}}}}}}}&0& \cdots &0
\end{array}} \right\}
\end{equation}
\begin{equation}\label{E.32}
{\overline {\bf{R}} _{\rm{F}}} = {{\bf{V}}_{\rm{f}}}{\overline {\bf{\Sigma }} _{\rm{f}}}{\bf{V}}_{\rm{f}}^{\rm{H}},{\overline {\bf{\Sigma }} _{\rm f}} = diag\left\{ {\begin{array}{*{20}{c}}
{\lambda _{{\rm{f}},1}^{ - 1}}& \cdots &{\lambda _{{\rm{f}},{\rho _{{\rm{f}}}}}^{ - 1}}&0& \cdots &0
\end{array}} \right\}
\end{equation}
where there exists ${\bf{R}}_{{\rm F}} {\rm{=}}  {\bf{F}}_{{\rm{L}},{s}}^{\rm{T}}{\bf{F}}_{{\rm{L}},{s}}^{\rm{*}}$.
 The rank of ${{\bf{R}}_i}$ and   ${{\bf{R}} _{\rm{F}}}$ are  respectively denoted by  $\rho_{i}$ and ${\rho _{{\rm{f}}}} = \min \left\{ {s,L} \right\}$.  To identify the two hypothesis, we build up the error decision function as
     \begin{equation}\label{E.34}
\Delta f  \buildrel \Delta \over = f\left( {{{\widehat {\bf{h}}}_{{\rm{B}},{\rm{L}}}}} \right) - f\left( {{{\widehat {\bf{h}}}_{{\rm{A}},{\rm{L}}}}} \right)
  \end{equation}
  where  the function $f$ satisfies $f\left( {\bf{r}} \right) = {\bf{r}}\left( {{\overline {\bf{R}} _1} \otimes {\overline {\bf{R}} _{\rm{F}}}} \right){{\bf{r}}^{\rm{H}}}$.  The function $f$ can be simplified by the following theorem:
\begin{theorem}
When ${N_{\rm{T}}} \to \infty $, the error decision function can be simplified as:
 \begin{equation}\label{E.35}
  \Delta f  = L\left\{ {{\rho _{\rm{1}}} - {\rm{Tr}}\left( {{{\bf{R}}_2}{{\overline {\bf{R}} }_1}} \right)} \right\}
  \end{equation}
 \end{theorem}
 \begin{IEEEproof}
See proof in Appendix~\ref{Theorem4}
\end{IEEEproof}
Examining this equation, we could find the pilot scheduling strategies of Ava across subcarriers do not affect the decision function.
In what follows, we try to further acquire the characteristic of $  \Delta f $ from the observation of  ${{\bf{R}}_1}$ and ${{\bf{R}}_2}$.

 \subsubsection{Hints Derived from  Spatial Correlation}
  \begin{algorithm}[!t]\label{CEAIE}
\caption{:Channel  Estimation and  Identification  Enhancement}
\begin{algorithmic}[1]
\STATE Identify whether attack happens through  the codewords derived by using inner product in~\cite{Xu}.
\STATE
 If  attack  happens, calculate the sample covariance matrix  ${{\bf{C}}_{{{\bf{Y}}_{\rm L}}}} = \frac{1}{{{N_{\rm{T}}}s}}{{\bf{Y}}_{\rm L}}{\bf{Y}}_{\rm L}^{\rm{H}}$. Derive the two  pilot signal vectors ${\bf{x}}_{{\rm{L,1}}}$ and ${\bf{x}}_{{\rm{L,2}}}$.   Calculate the weight matrices and  finally derive the FS channel estimations using Eq.~(\ref{E.27}).  If no attack happens, just use LS estimator to get FS channels.
\STATE
If no attack  happens, directly derive CIR estimation using estimated FS channels,  otherwise,  calculate $\Delta f $ using Eq.~(\ref{E.34}). According to Theorem 4, if $\Delta f>0 $,  ${\widehat {\bf{h}}_{{\rm{B}},{\rm{L}}}}$ serves as  the true estimated FS channel of Bob for further CIR estimating, otherwise $\Delta f<0 $,   ${\widehat {\bf{h}}_{{\rm{A}},{\rm{L}}}}$ does.  When $\Delta f=0 $, an identification error happens and the  reliability breaks down.
\end{algorithmic}
\end{algorithm}

 The authors in~\cite{Adhikary} pointed out that the set of eigenvalues of ${{\bf{R}}_i}$ and the set of uniformly spaced samples $\left\{ {S_{i}\left( {{n \mathord{\left/
 {\vphantom {n {{N_{\rm{T}}}}}} \right.
 \kern-\nulldelimiterspace} {{N_{\rm{T}}}}}} \right):n = 0, \ldots ,{N_{\rm{T}}} - 1} \right\}$ are asymptotically equally distributed, i.e., for any continuous function $f\left( x \right)$. The function $S_{i}\left( x \right)$ over $x \in \left[ { - \frac{1}{2},\frac{1}{2}} \right]$ satisfies:
${S_i}\left( x \right) = \frac{1}{{2\Delta }}\sum\limits_{0 \in \left[ {D\sin \left( {{\theta _i} - \Delta } \right) + x,D\sin \left( {{\theta _i} + \Delta } \right) + x} \right]} {\frac{1}{{\sqrt {{D^2} - {x^2}} }}}$. And the channel covariance eigenvectors ${{\bf{U}}_i}$, i.e., $N_{{\rm T}} \times {\rho _{{i}}}$ matrix ${{{\overline {\bf{U}} }_i}}$, can be approximated with a submatrix of the DFT matrix ${\bf{F}}$ in the following sense:
$\mathop {\lim }\limits_{N_{{\rm T}} \to \infty } \frac{1}{{{N_{\rm{T}}}}}\left\| {{{\overline {\bf{U}} }_i}\overline {\bf{U}} _i^{\rm{H}} - {{\bf{F}}_{{{\cal S}_i}}}{\bf{F}}_{{{\cal S}_i}}^{\rm{H}}} \right\|_{\rm{F}}^2 = 0,i = 1,2$
where ${{\bf{F}}_{{{\cal S}_i}}} = \left( {{{\bf{f}}_n}:n \in {{\cal J}_{{{\cal S}_i}}}} \right)$ with
${{\cal J}_{{{\cal S}_i}}} = \left\{ {n,\left[ {{n \mathord{\left/
 {\vphantom {n {{N_{\rm{T}}}}}} \right.
 \kern-\nulldelimiterspace} {{N_{\rm{T}}}}}} \right] \in {{\cal S}_i},n = 0, \ldots ,{N_{\rm{T}}} - 1} \right\}$
  Here,  ${{\cal S}_i}$ denotes the support of ${S}_{i}\left( x \right)$.
Backing to the Eq.~(\ref{E.35}),  the trace function satisfies  ${\rm{Tr}}\left( {{{\bf{R}}_2}{{\overline {\bf{R}} }_1}} \right) \le {\rm{Tr}}\left( {{{\bf{\Lambda }}_2}{\bf{U}}_2^{\rm{H}}{{\bf{U}}_1}{{\overline {\bf{\Lambda }} }_1}} \right)={\rm{Tr}}\left( {{{\bf{\Lambda }}_{2,{\rm{p}}}}\overline {\bf{U}} _2^{\rm{H}}{{\overline {\bf{U}} }_1}{{\overline {\bf{\Lambda }} }_{1,{\rm{p}}}}} \right)$ where ${{\bf{\Lambda }}_{i,{\rm{p}}}} $ and ${\overline {\bf{\Lambda }} _{i,{\rm{p}}}}$ are respectively defined by ${{\bf{\Lambda }}_{i,{\rm{p}}}} = diag\left\{ {\begin{array}{*{20}{c}}
{{\lambda _{i,1}}}& \cdots &{{\lambda _{i,{\rho _i}}}}
\end{array}} \right\}$ and ${\overline {\bf{\Lambda }} _{i,{\rm{p}}}} = diag\left\{ {\begin{array}{*{20}{c}}
{\lambda _{i,1}^{ - 1}}& \cdots &{\lambda _{i,{\rho _i}}^{ - 1}}
\end{array}} \right\}$. As previously discussed, we  approximate  ${\overline {\bf{U}} _2^{\rm{H}}{{\overline {\bf{U}} }_1}}$ using  ${\bf{F}}_{{{\cal S}_2}}^{\rm{H}}{{\bf{F}}_{{{\cal S}_1}}}$ and define  ${{\cal S}_1} \cap {{\cal S}_2} = {{\cal S}_3}$.  We  then discuss the influence of  ${{\cal S}_3}$ on ${\rm{Tr}}\left( {{{\bf{\Lambda }}_{2,{\rm{p}}}}\overline {\bf{U}} _2^{\rm{H}}{{\overline {\bf{U}} }_1}{{\overline {\bf{\Lambda }} }_{1,{\rm{p}}}}} \right)$.  When ${{\cal S}_3} = \emptyset $, we can have ${\rm{Tr}}\left( {{{\bf{R}}_2}{{\overline {\bf{R}} }_1}} \right)=0$. When ${{\cal S}_3}\ne \emptyset$, we assume ${{\cal S}_3} = {\cal P}_{a}$ and have
  \begin{equation}\label{E.39}
{\rm{Tr}}\left( {{{\bf{\Lambda }}_{2,{\rm{p}}}}\overline {\bf{U}} _2^{\rm{H}}{{\overline {\bf{U}} }_1}{{\overline {\bf{\Lambda }} }_{1,{\rm{p}}}}} \right) \le \sum\limits_{j = 1}^a {\frac{{{\lambda _{2,{i_j}}}}}{{{\lambda _{1,{i_j}}}}}}
  \end{equation}
 This is because  the eigenvectors labeled by the indexes out of  the interacted set ${{\cal S}_3}$ are mutually orthogonal~\cite{Adhikary}.
Then we have the following theorem:
  \begin{theorem}
 When ${N_{\rm{T}}} \to \infty $,  there always exists $\sum\limits_{j = 1}^a {\frac{{{\lambda _{2,{i_j}}}}}{{{\lambda _{1,{i_j}}}}}} =  a$.  If ${\theta _1} \ne {\theta _2}$,   there must exist  $a< {\rho _1} $ and $ \Delta f > 0$. Otherwise if $ {\theta _1} = {\theta _2}$, there must exist  $a={\rho _1}$ and $ \Delta f =0$.
\end{theorem}
\begin{IEEEproof}
See proof in Appendix~\ref{Theorem5}
 \end{IEEEproof}
 \subsubsection{IEP Reduction}
  Inspired by the above result, we know that the identification error  happens only  when  ${\theta _1} = {\theta _2}$.
  \begin{theorem}
 Under the assumption of  mean AoA obeying  continuous probability distribution (CPD), the IEP ${{P}_{\rm{I}}} $ in Eq. (9)  is updated  to be  zero.
 Under the assumption of  mean AoA obeying  distrete probability distribution (DPD), for instance, uniform distribution with interval length $K$,   the IEP ${{P}_{\rm{I}}} $  in Eq. (9) is updated  to be $\frac{{{P_{\rm{I}}}}}{K}$.
  \end{theorem}
 The proof is institutive. Therefore, the IEP can be seriously reduced and reliability  is thus significantly enhanced  under hybrid attack environment. Finally, we give the overall process of  channel estimation and  identification enhancement  in Algorithm 2.
%\begin{figure}
%\centering \includegraphics[width=0.75\linewidth]{PD_PF.eps}
%\caption{ Performance of ERD for each single subcarrier with various antenna configurations.}
%\vspace{-10pt}
%\label{DP}
%\end{figure}

\begin{figure}
\centering \includegraphics[width=0.75\linewidth]{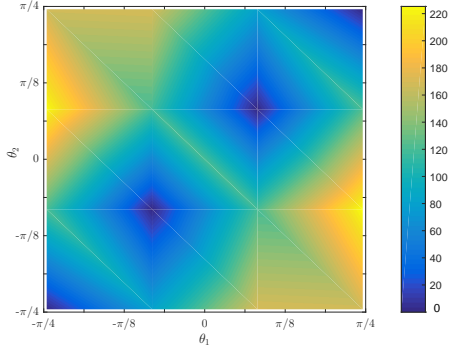}
\caption{Strength  of $\Delta f$ versus $\theta_{i},i=1,2$ with $N_{\rm T}=100$.}
\vspace{-10pt}
\label{Identification}
\end{figure}
\section{Numerical Results}
\label{NR}
In this section, we further  carry out  the performance  evaluation concerning above techniques mentioned.

In this part, we aim to verify the feasibility of Theorem 4 through simulations shown in Fig.~\ref{Identification} where  the strength  of $\Delta f $ is plotted  against  $\theta_{i}, i=1,2$ by configuring $N_{\rm T}=100$ and  $K=5$. In this simulation, we consider  that the candidate samples of  discrete mean AoAs lie within the  set $\left\{ { - \frac{\pi }{4}, - \frac{\pi }{7},0, - \frac{\pi }{7}, - \frac{\pi }{4}} \right\}$. Based on the estimation in Eq.~(\ref{E.27}) and the correlation model in Eq.~(\ref{E.5}), we derive the corresponding examples of   $\Delta f $.  As we can  see,  the identification error happens when $\Delta f =0$, that is, $\theta_{1}=\theta_{2}$.  In this sense, we could envision  that the IEP is zero under the assumption of the mean AoA with CPD.

For the sake of a comprehensive analysis, we consider the DPD model for mean AoA and further simulate the IEP performance in Fig.~\ref{IEP}. The mean AoA is  discretely and  uniformly  distributed in a length-$K$ interval.  As shown in this figure,  the performance  of  IEP is plotted versus the length of $N_{\rm B}$  under different number  $L$ of channel taps.   We consider $L$ to be from 7 to 13 and $K$ to be 20. $k$, related to $N_{\rm B}$, satisfies $N_{\rm B}= 2k+1$.    As we can see,  even with small subcarrier overheads, that is, $N_{\rm B}$ is small,  the IEP can be low and our architecture has  a very reliable  performance  guarantee.
 Moreover, we can find that when  $L$ is low, such as $L=7$,   the IEP has a maximum value after which IEP decreases with the increase of  $N_{\rm B}$.  With the increase of $L$, IEP  decreases monotonically with $N_{\rm B}$.  Furthermore, the initial value of $L$ determines  the   upper bound  IEP can achieve.  With the increase of $L$, the upper bound decreases. For instance,  the upper bound of $P_{\rm I}$  achieves as low as $10^{-3.3}$ at $k=80$ when $L$ is equal to 9. In this case, the number of occupied subcarriers satisfies $N_{\rm B}=161$.
\begin{figure}
\centering \includegraphics[width=0.75\linewidth]{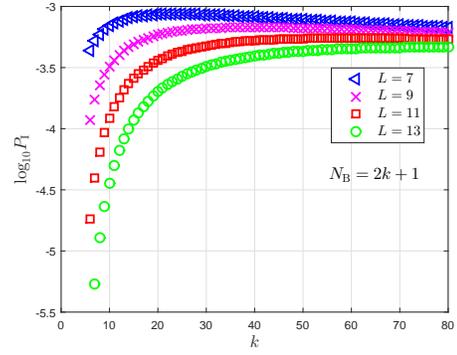}
\caption{Performance of IEP versus $N_{\rm B}$  under different $L$.}
\vspace{-10pt}
\label{IEP}
\end{figure}

Finally, we simulate the performance  of  channel estimation in Fig.~\ref{NMSE} in which   the  NMSE is plotted  versus SNR of Bob under different number of antennas.  $L$ and $N_{\rm B}$ are respectively  configured to be 6 and  256. Here,  we  consider the  estimation shown in Eq.~(\ref{E.27})  and assume perfect identification  for attacks. We do not consider the case where there is no attack since in this case LS estimator is a natural choice. For the simplicity of  comparison, we only present the channel estimation  under PTS attack because  the estimation error floor  under PTN and PTJ attack can be easily understood to be very high. The binned scheme proosed in~\cite{Shahriar2} is simulated as an another  comparison scheme.   As we can see,  PTS attack, if happens,  causes a high-NMSE floor on  CIR estimation  for Bob. This phenomenon can also be seen in the binned scheme. However, the estimation in our proposed framework  breaks down this floor and its NMSE  gradually decreases with the increase of transmitting antennas.  Also, we consider perfect MMSE to be  a performance benchmark for which  perfect pilot tones, including Ava's pilot tones,  are assumed to be  known by Alice. We find that  the NMSE brought in our scheme  gradually approaches the level under perfect MMSE with the increase of antennas. That's because the  estimator highly relies on the statistical property of  ${{\bf{C}}_{{{\bf{Y}}_{\rm L}}}}$, determined by the number of antennas.

\begin{figure}
\centering \includegraphics[width=0.75\linewidth]{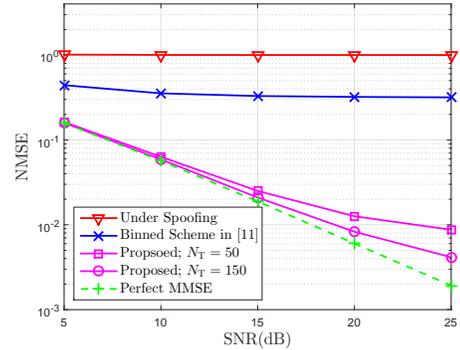}
\caption{NMSE of CIR estimation versus SNR under different number of antennas.}
\vspace{-10pt}
\label{NMSE}
\end{figure}
\section{Conclusions}
\label{Conclusions}
This paper investigated the issue of pilot-aware attack on the uplink channel training process in large-scale MISO-OFDM systems. We proposed a secure ICCB uplink training architecture in which  pilot tones, usually exposed in public, are now enabled  to be shared between legitimate transceiver pair under hybrid attack environment. We developed a novel coding theory to support  and secure this pilot sharing process, and found an optimal code rate to finally provide  the well-imposed CIR estimation. Theoretically, we verified an important fact that this architecture could perfectly secure  pilot sharing  against the attack if the CPD  model of mean AoA was considered. In practical scenarios with DPD model of mean AoA, this architecture  could also bring  a high-reliability and high-precision CIR estimation.
\section{Appendix}
\subsection{Proof of Theorem 1}
\label{Theorem1}
Since codewords in this constant-weight code are constrained to be  with same and  fixed length,  the number of overlapping digits achieves its minimum only when the zero digits of each codeword  are  fully occupied. In this case, the remanent digits, i.e., the overlapping digits,  account for  ${2w-N_{\rm B}}$ which should be  equal to $s$ and less than $w$.  Therefore, we can prove the theorem.
\subsection{Proof of Theorem 2}
\label{Theorem2}
To guarantee well-imposed CIR estimation, there should be $s\ge L$, that is, ${{\bf{F}}_{{\rm{L}}, s}}$ is with full column rank. However, the increase of $s$ will reduce the code rate since the function of $C$ decreases with  $s$. Therefore, the optimal code rate is proved.

Considering  the hybrid attack,  we know that their exists the  possibility  of  $2^{N_{\rm B}}$   codewords to appear.  Two interpreted codewords derived by inner-product operation in ~\cite{Xu},  if satisfying  the same weight constraint, will confuse Alice. In this case,  each assumption is  decided with the probability of $0.5$.   The possible number of codewords that satisfy  this condition is equal to   ${ {\frac{{{N_{\rm B}}!}}{{\left( {\frac{{{N_{\rm B}} + L}}{2}} \right)!\left( {\frac{{{N_{\rm B}} - L}}{2}} \right)!}}}}$.  One exception is when  the codeword of Ava is identical to that of Bob. In this case, the codeword can be uniquely determined.  Finally,  there exists the  possibility of   ${ {\frac{{{N_{\rm B}}!}}{{\left( {\frac{{{N_{\rm B}} + L}}{2}} \right)!\left( {\frac{{{N_{\rm B}} - L}}{2}} \right)!}}}-1}$ codewords that could cause identification errors. Then the ultimate IEP can be proved.
\subsection{Proof of Proposition 2}
\label{appendices_pro2}
Take Bob for example, we can derive the error of MMSE based estimation  as $\varepsilon _{\rm{B}}^2 = T\left( {1 - T{\bf{X}}_{{\rm{L}},{\rm{1}}}^{\rm{H}}{\bf{C}}_{{{\bf{Y}}_{\rm{L}}}}^{ - 1}{{\bf{X}}_{{\rm{L}},{\rm{1}}}}} \right)$.

 ${{\bf{C}}_{{{\bf{Y}}_{\rm L}}}}$  is  transformed into
${{\bf{C}}_{{{\bf{Y}}_{\rm L}}}}  \xlongrightarrow[{N_{\rm T}} \to \infty]{ \rm{a.s.}} \frac{1}{{{N_{\rm{T}}}s}}{{\bf{X}}_{{\rm{L}}}}{{\bf{R}}_{\rm C}}{\bf{X}}_{{\rm{L}}}^{\rm{H}} + {\sigma ^2}{{\bf{I}}_2}$ using  asymptotic approximation~\cite{Hoydis}. Here,  the $2\times 2$  matrix ${{\bf{R}}_{\rm C}}$  satisfies ${{\bf{R}}_{\rm{C}}} = {\rm{diag}}\left\{ {\begin{array}{*{20}{c}}
{{\rm{Tr}}\left( {{{\bf{R}}_{{1}}}} \right){\rm{Tr}}\left( {{{\bf{R}}_{\rm{F}}}} \right)}&{{\rm{Tr}}\left( {{{\bf{R}}_{{2}}}} \right){\rm{Tr}}\left( {{{\bf{R}}_{\rm{F}}}} \right)}
\end{array}} \right\}$.
Therefore, we can derive  $\varepsilon _{\rm{B}}^2 =T\left\{ {1 - {\bf{X}}_{{\rm{L}},{\rm{1}}}^{\rm{H}}{{\left( {{{\bf{X}}_{\rm{L}}}{\bf{X}}_{\rm{L}}^{\rm{H}}} \right)}^{ - 1}}{{\bf{X}}_{{\rm{L}},{\rm{1}}}}} \right\}$  at high SNR region.
In the same way, we can derive  $\varepsilon _{\rm{A}}^2 =T\left\{ {1 - {\bf{X}}_{{\rm{L}},{\rm{2}}}^{\rm{H}}{{\left( {{{\bf{X}}_{\rm{L}}}{\bf{X}}_{\rm{L}}^{\rm{H}}} \right)}^{ - 1}}{{\bf{X}}_{{\rm{L}},{\rm{2}}}}} \right\}$.
After calculating  the matrix inverse and performing matrix multiplication,  we can finally verify $ \varepsilon _{\rm{B}}^2 = \varepsilon _{\rm{A}}^2 $. This completes the proof.
\subsection{Proof of Theorem 3}
\label{Theorem4}
Thanks to ${\widehat {\bf{h}}_{\rm{B,L}}} = {{\bf{h}}_{\rm{B,L}}} - \varepsilon _{{\rm{B}}}{\bf{h}}$, $ f\left( {{{\widehat {\bf{h}}}_{{\rm{B}},{\rm{L}}}}} \right)$ can be expressed as:
${f\left( {{{\widehat {\bf{h}}}_{{\rm{B}},{\rm{L}}}}} \right) = \left( {{{\bf{h}}_{{\rm{B}},{\rm{L}}}} - {\varepsilon _{\rm{B}}}{\bf{h}}} \right)\left( {{{\overline {\bf{R}} }_1} \otimes {{\overline {\bf{R}} }_{\rm{F}}}} \right){{\left( {{{\bf{h}}_{{\rm{B}},{\rm{L}}}} - {\varepsilon _{\rm{B}}}{\bf{h}}} \right)}^{\rm{H}}}}$
then this equation can be expanded into $ f\left( {{{\widehat {\bf{h}}}_{{\rm{B}},{\rm{L}}}}} \right) = {f_1} - 2{f_2} + {f_3}$
where ${f_1} = {{\bf{h}}_{{\rm{B}},{\rm{L}}}}\left( {{{\overline {\bf{R}} }_1} \otimes {{\overline {\bf{R}} }_{\rm{F}}}} \right){\bf{h}}_{{\rm{B}},{\rm{L}}}^{\rm{H}}$ and ${f_2} = {\varepsilon _{\rm{B}}}{{\bf{h}}_{{\rm{B}},{\rm{L}}}}\left( {{{\overline {\bf{R}} }_1} \otimes {{\overline {\bf{R}} }_{\rm{F}}}} \right){\bf{h}}$, ${f_3} = \varepsilon _{\rm{B}}^2{\bf{h}}\left( {{{\overline {\bf{R}} }_1} \otimes {{\overline {\bf{R}} }_{\rm{F}}}} \right){\bf{h}}$.
By using the asymptotic approximation for each term, we can have
$f\left( {{{\widehat {\bf{h}}}_{{\rm{B}},{\rm{L}}}}} \right)  \xlongrightarrow[{N_{\rm T}} \to \infty]{ \rm{a.s.}}  {\rho_{1}}L{\rm{ + }}\varepsilon _{\rm{B}}^2{\rm{Tr}}\left( {{{\overline {\bf{R}} }_1} \otimes {{{\overline {\bf{R}} }_{\rm{F}}}}} \right)$
In the same way, we can simplify  the function $f\left( {{{\widehat {\bf{h}}}_{{\rm{A}},{\rm{L}}}}} \right)$ as:
$f\left( {{{\widehat {\bf{h}}}_{{\rm{A}},{\rm{L}}}}} \right)  \xlongrightarrow[{N_{\rm T}} \to \infty]{ \rm{a.s.}}  L{\rm{Tr}}\left( {{{\bf{R}}_2}{{\overline {\bf{R}} }_1}} \right){\rm{ + }}\varepsilon _{\rm{A}}^2{\rm{Tr}}\left( {{{\overline {\bf{R}} }_1} \otimes {\overline {\bf{R}} }_{\rm{F}} } \right)$
 As indicated in Proposition 2, there exists $\varepsilon _{\rm{B}}^2=\varepsilon _{\rm{A}}^2$. By comparing $f\left( {{{\widehat {\bf{h}}}_{{\rm{B}},{\rm{L}}}}} \right)$ and $f\left( {{{\widehat {\bf{h}}}_{{\rm{A}},{\rm{L}}}}} \right)$, we  can complete the proof.
 \subsection{Proof of Theorem 4}
 \label{Theorem5}
 First, we will prove $\sum\limits_{j = 1}^a {\frac{{{\lambda _{2,{i_j}}}}}{{{\lambda _{1,{i_j}}}}}}  = a$. As shown in~\cite{Adhikary}, the empirical CDF of eigenvalues of ${\bf {R}}_{i}$ can be asymptotically  approximated  using the collection of samples from  $\left\{ {{S_i}\left( {\left[ {{n \mathord{\left/
 {\vphantom {n {{N_{\rm{T}}}}}} \right.
 \kern-\nulldelimiterspace} {{N_{\rm{T}}}}}} \right]} \right),n = 0, \ldots ,{N_{\rm{T}}} - 1} \right\}$. Therefore, the eigenvalues of different individuals, if overlapping at the same location $n$, can be approximated with the same eigenvalues.

 Then we prove that there must $a < {\rho _1}$. Examining $\left[ {{\theta _2} - {\Delta },{\theta _2} + {\Delta }} \right]$ and  $\left[ {{\theta _1} - {\Delta },{\theta _2} + {\Delta }} \right]$ we found  if $ {\theta _1} \ne {\theta _2}$ is satisfied, there must exist $a < {\rho _1}$ since $\left[ {{\theta _2} - {\Delta },{\theta _2} + {\Delta }} \right]$ must have non-empty intersection  with  $\left[ {{\theta _1} - {\Delta },{\theta _1} + {\Delta}} \right]$.  In this case, the  number of elements in ${{\cal S}_{3}}$  is reduced to be smaller than that ${\rho _1}$.
 Now we turn to the case ${\theta _1} = {\theta _2}$ in which we easily have ${{\bf{R}}_1} = {{\bf{R}}_2}$ and therefore the theorem is proved.

%\begin{IEEEbiography}[{\includegraphics[width=1in,height=1.25in,clip,keepaspectratio]{sunli.eps}}]
%\end{IEEEbiography}
% that's all folks
\end{document}